\documentclass[onecolumn,secnumarabic,amssymb, aps, pra,notitlepage,showpacs]{revtex4-1}
\usepackage{amsmath,amsfonts,amssymb,amsthm}
\usepackage[colorlinks]{hyperref} 
\usepackage{dsfont}
\usepackage{comment}
\usepackage{bbm}
\usepackage{MnSymbol}
\usepackage{mathdots}
\usepackage[toc,page]{appendix}
\usepackage{float}
\usepackage{xspace}
\usepackage[caption=false]{subfig}
\usepackage{graphicx}
\usepackage{color}
\usepackage{clrscode}
\usepackage{epsfig}                  
\usepackage{enumerate}

\newtheorem{theorem}{Theorem}
\newtheorem{lemma}{Lemma}[section]

\newtheorem{corollary}{Corollary}

\def\be{\begin{equation}}
\def\ee{\end{equation}}
\def\ba{\begin{eqnarray}}
\def\ea{\end{eqnarray}}

\def\la{\langle}
\def\ra{\rangle}

\def\m{\mu}

\def\la{\langle}
\def\ra{\rangle}

\DeclareMathOperator{\Tr}{Tr}

\DeclareMathOperator*{\sth}{subject\ to}
\newcommand{\ket}[1]{|{#1}\rangle}
\newcommand{\kett}[1]{|{#1}\rrangle}
\newcommand{\bra}[1]{\langle{#1}|}
\newcommand{\braa}[1]{\llangle{#1}|}

\newcommand{\brakett}[2]{\llangle{#1}|{#2}\rrangle}

\newcommand{\pr}[1]{\otimes{#1}}

\begin{document}

\title{No-iteration of unknown quantum gates}
\author{Mehdi Soleimanifar}%
\email{mehdi.soleimanifar@gmail.com}
\author{Vahid Karimipour}%
\email{vahid@sharif.edu}
\affiliation{Department of Physics, Sharif University of Technology, Tehran, Iran}

\begin{abstract}
We propose a new no-go theorem by proving the impossibility of constructing a deterministic quantum circuit that iterates a unitary oracle by calling it only once. Different schemes are provided to bypass this result and to approximately realize the iteration. The optimal scheme is also studied. An interesting observation is that for large number of iterations, a trivial strategy like using the identity channel has the optimal performance, and preprocessing, postprocessing, or using resources like entanglement does not help at all. Intriguingly, the number of iterations, when being large enough, does not affect the performance of the proposed schemes. 
\end{abstract}
\pacs{03.67.-a, 03.67.Lx, 03.67.Ac}
\maketitle
\section{Introduction}
\label{sec:intro}
\indent
No-go theorems play a major role in quantum information science. The impossibility of perfect cloning of an unknown pure state, the \emph{no-cloning theorem}, is one of the striking features of quantum mechanics \cite{Wootters}. This no-go result is fundamental to key distribution \cite{Scarani}, quantum secret sharing \cite{Hillery}, and quantum error correction \cite{Bennett}. A similar no-go theorem is valid for cloning of an \emph{unknown quantum gate} from one to two copies \cite{unitary-cloning}; that is to say, given a set of distinct states $\bigotimes_{i=1}^m \ket{\psi_i}$ and an unknown unitary channel $\mathcal{U}$, it is impossible to prepare $\bigotimes_{i=1}^mU\ket{\psi_i}$ by a quantum circuit that uses $\mathcal{U}$ only once. This result has implications in cryptographic protocols where the secret is encoded in unitary transformations instead of quantum states, e.g., an alternative version of BB84 protocol where Alice uses two orthogonal bases of unitary transformations instead of states \cite{unitary-cloning}. \\

The no-cloning of states is about the impossibility of realizing a specific transformation of \emph{states}, while the no-cloning of gates is about a transformation of \emph{unitary channels}. Other examples of no-go theorems on transformations of quantum channels are: The impossibility of realizing the \emph{switch} circuit defined by $\mathcal{Z}(\mathcal{V},\mathcal{W})=\ket{0}\bra{0}\pr{\mathcal{VW}}+\ket{1}\bra{1}\pr{\mathcal{WV}}$, in which a pair of input unitary blackboxes $\mathcal{V}$ and $\mathcal{W}$ are connected in two different orders conditioned on the value of an input bit \cite{causal-order}. By generalizing the conventional quantum circuit model to bypass this no-go result, a computational advantage can be obtained \cite{computational-advantage}. Another example is the no-go theorem on controlling a unitary gate given as a blackbox discussed in \cite{Brukner,Friis,Bisio}, or failure of programming a quantum gate array $\mathcal{G}$ that deterministically implements the unitary operation $\mathcal{U}$ determined by the quantum program $\ket{P_U}$ or strictly speaking $\mathcal{G}(\ket{\psi}\pr{\ket{P_U}})=U\ket{\psi}\pr{\ket{P'_U}}$ \cite{no-programming}. \\

In this paper, we introduce and investigate a new no-go theorem on \emph{iterations of an unknown quantum gate}. The iteration of a unitary gate is widely used in quantum algorithms. Quantum search algorithms like Grover algorithm \cite{Grover} or quantum random walk search algorithm \cite{random-walk-search} are based on the repetition of a unitary oracle. They use iterations of the oracle to amplify the amplitude of a desired state in a superposition of states \cite{Amplification}. Quantum phase estimation \cite{phase-estimation} is another algorithm in which successive iterations of a unitary gate are used to generate states appropriate for  an inverse quantum Fourier transform. These algorithms are bases of other quantum computations like order finding \cite{phase-estimation}, integer factorization and discrete logarithms \cite{Shor} or the collision problem \cite{collision}.\\

The question we try to answer is whether it is possible to avoid iterations of a unitary oracle using a deterministic quantum circuit. One possible scenario for doing this is that an apparatus called \emph{gate iterator} of the $n$'th order ($n\in \mathbb{N}/\{1\}$), denoted by $\mathsf{Iter_n}$, takes a unitary oracle $\mathcal{U}$, an arbitrary state $\ket{\psi}$ and the state of the rest of the world $\ket{0}$ as inputs, and by calling $\mathcal{U}$ once, gives $U^n\ket{\psi}$ as the output. The state $\ket{0}$ may also change to another state $\ket{0^{\prime}}$ at the end, see Fig. \ref{fig:ckt-1}. In a more general scenario, the input system could be mixed and the output state be entangled with the ancillary system.\\

\begin{figure}
\centering 
\includegraphics[width=.3\textwidth]{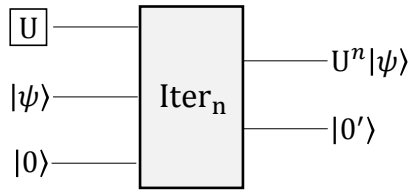}
\caption{In a possible scenario, the gate iterator apparatus, $\mathsf{Iter_n}$, takes $\mathcal{U}$, $\ket{\psi}$ and $\ket{0}$ as inputs and gives $U^n\ket{\psi}$ as the output.}
\label{fig:ckt-1} 
\end{figure}

We prove that it is impossible to realize $\mathsf{Iter}_n$ which consists of a deterministic quantum circuit. We first consider the most general circuit for doing the iteration consists of a \emph{preprocessing} and a \emph{postprocessing} channel. We then show that such a procedure contradicts the linearity of quantum mechanics. \\

Although it is not possible to construct $\mathsf{Iter_n}$ perfectly, it is natural to ask for strategies that realize it in an approximate way. We propose such schemes and then, by using the notion of \emph{fidelity} as a figure of merit, we investigate their performance and how it scales with the number of iterations $n$, and the dimensionality of the state $d$. We also address the problem of finding the optimal iterator. We show that the optimal fidelity is the answer of a semidefinite programming. We solve this problem numerically for $d=2, 3$, see Fig. \ref{fig:data_plot_3}.\\

As we will see, the approximate realization of a gate iterator has interesting characteristics. We show that in all strategies, including the optimal, by increasing the dimension of the system $d$, the fidelity decreases. Intriguingly, the fidelity reaches a constant value by increasing $n$, so when we are allowed to query a unitary oracle only once, there is no difference in quality of high-order imperfect iterations of that. Another interesting observation is that when $n>d$, a trivial strategy like approximating $\mathcal{U}^n$ by $\mathcal{U}$ or even by the identity channel has the optimal performance, and preprocessing, postprocessing, or using resources like entanglement does not help at all. These results are depicted in Fig. \ref{fig:plot1}, \ref{fig:plot_data_2} and \ref{fig:data_plot_3}. \\
 
An anticipated result of our no-go theorem is that when the oracle is completely unknown, the only way to perform iterations of that, seems to be calling the oracle each time and give the state to that and do this repeatedly. For a large number of iterations, this makes the algorithm inefficient. In fact, one way of comparing the complexity of different algorithms is to count the necessary number of querying within the program \cite{querying}. \\

The rest of the paper is organized as follows: In the next section, notations and some basic definitions are presented, and a convenient figure of merit is introduced to measure how good a quantum circuit approximates a given unitary gate. In Sec. \ref{sec:no-repetition}, we prove the no-iteration theorem for an arbitrary order $n$. The fidelity and details of the \emph{random guess} and \emph{measure-and-prepare} strategies are discussed in Sec. \ref{sec:quantum circuit of a repeater} and \ref{sec:random-guess}. Then, two trivial but important methods are introduced in Sec. \ref{sec:idn strategies}, and the optimum fidelity is obtained numerically in Sec. \ref{sec:Optimal strategies}. Finally, we conclude the paper in Sec. \ref{sec:Conclusion}. 
\section{Notations and conventions}\label{Notations and conventions}
In this section, we gather some well-known facts which we will frequently use in the sequel. We denote the complex $d-$dimensional Hilbert space by ${\cal H}_d$,  and the linear space of operators acting on it by $L({\cal H}_d)$ and the set of density matrices by $D({\cal H}_d)$.  A basis for ${\cal H}_d$ is denoted by $\{|j\ra:  j=1,2,\dots,n\}$ and any linear operator $V$ in $L({\cal H}_d)$ is expanded as $V=\sum_{j,k}V_{jk}|j\ra\la k|$.  
A correspondence between this operator and a vector $\kett{V} \in \mathcal{H}_d\pr{\mathcal{H}_d}$ can be established by defining
\begin{align}
\kett{V}:=\frac{1}{\sqrt{d}} \sum_{j,k} V_{jk}\ket{j}\ket{k},
\end{align}
where $\kett{V}$ is called the \emph{vectorized form} of the operator $V$. 
Therefore, the maximally entangled state $|\phi^+\ra:=\frac{1}{\sqrt{d}}\sum_{j}|j\ra |j\ra$ is the vectorized form of the identity operator, i.e., $|\phi^+\ra=\kett{\mathds{1}}$. The inner product between two operators $A$ and $B$ defined as $\Tr(A^\dagger B)$ can equally be written as the ordinary vector product of their vectorized form, that is $\Tr(A^{\dagger}B)=\brakett{A}{B}$. Finally, we note that a vector $\kett{V}$ can be prepared by performing $V$ on the maximally entangled state $\kett{\mathds{1}}$:
\begin{align}\label{eq:-8}
\kett{V}=V\pr{\mathds{1}}\ \kett{\mathds{1}}.
\end{align}

The Choi operator $R_{\mathcal{T}}$ associated with a quantum channel $\mathcal{T}:D(\mathcal{H}_d)\rightarrow D(\mathcal{K}_d)$ is defined on $\mathcal{K}_d\pr{\mathcal{H}_d}$ by
\begin{align}\label{eq:3}
R_{\mathcal{T}}:=(\mathcal{T}\pr{\mathcal{I}})\ (\kett{\mathds{1}}\braa{\mathds{1}}),
\end{align}
where  $\mathcal{I}$ is the identity channel. Obviously, we have $R_{\mathcal{I}}:=\kett{\mathds{1}}\braa{\mathds{1}}$, that is to say, the Choi operator of the identity channel is the Bell state. \\

A unitary quantum channel (quantum gate) $\mathcal{U}$ is defined as
\begin{align}\label{unitary channel}
\mathcal{U}(\rho):=U\rho U^{\dagger},
\end{align}
that according to Eq. (\ref{eq:-8}), its Choi operator is the pure state $R_{\mathcal{U}}=\kett{U}\braa{U}$. \\

When we want to evaluate the performance of a process $\mathcal{E}$ that approximates a gate $\mathcal{U}$, we need to introduce a figure of merit. The fidelity between two channels $\mathcal{G}$ and $\mathcal{E}$ is defined to be the state fidelity between the Choi operators of these channels \cite{Raginsky}:
\begin{align}
\mathcal{F}(\mathcal{G},\mathcal{E}):= \left(\Tr \left(\sqrt{\sqrt{R_{\mathcal{G}}}R_\mathcal{E}\sqrt{R_{\mathcal{G}}}} \right) \right)^2,
\end{align}
which reduces to the following when one of them is a unitary channel of the form (\ref{unitary channel})
\begin{align}\label{eq:-10}
\mathcal{F}(\mathcal{U}, \mathcal{E})=\braa{U} R_{\mathcal{E}}\kett{U}.
\end{align}

Now, we assume that instead of a single gate, a specific set of gates $S$, consists of a finite or infinite collection of unitary gates, are to be approximated with a process $\mathcal{E}$, and each gate $U\in S$ occurs with probability $P(U)$. The input of $\mathcal{E}$ is a given $U\in S$ and the output is ${\mathcal E}_{U}$. Then, a figure of merit that determines the performance of process $\mathcal{E}$ is given by:
\begin{align} \label{eq:-17}
F(\mathcal{E}):=\int dU P(U) \mathcal{F}(\mathcal{U},\mathcal{E}_{U}).
\end{align}
 Here, $dU$ is an invariant Haar measure, that is $d(UV)=d(VU)=dU,\ \forall V\in \mathbb{U}(d)$. When $S$ is the unitary group $\mathbb{U}(d)$, and gates are chosen uniformly, $P(U)=1,\ \forall U\in \mathbb{U}(d)$.
  
\section{No-iteration of unknown quantum gates}
\label{sec:no-repetition}
We now prove the impossibility of implementing $\mathsf{Iter}_n$. We call this no-go result, the \emph{no-iteration of unknown quantum gates} and provide two proofs for that. One, for the case that the output states are product states, is based on the linearity of the quantum circuit that implements $\mathsf{Iter_n}$, and a more general proof, available in Appendix \ref{app:-1}, is a corollary of a lower bound on the performance of quantum search algorithms. As another confirmation for the validity of this theorem, the optimum fidelity for approximating $\mathsf{Iter_n}$ is obtained numerically for $d=2, 3$ in Sec. \ref{sec:Optimal strategies}, and as expected, it is less than $1$.
\begin{theorem}\label{thm: no-rep}
The universal deterministic gate iterator of order $n$, $\mathsf{Iter_n}$, cannot be implemented perfectly.
\end{theorem}

\begin{proof}[\textbf{Proof}]
The most general quantum circuit that uses a single copy of $\mathcal{U}$ to implement $\mathcal{U}^n$ is depicted in Fig. \ref{fig:ckt} \cite{Arch, Supermaps}. $\mathcal{A}_n$ and $\mathcal{B}_n$ are preprocessing and postprocessing gates respectively, and $\ket{0}$ shows the ancillary system. This circuit transforms inputs to $B_n(U\otimes\mathds{1})A_n\ \ket{\psi}\otimes \ket{0}$. In this proof, it is assumed that input states are pure and output states are product states (this is relaxed in the alternate proof, see Appendix \ref{app:-1}), so the output is of the form $U^n \ket{\psi}\otimes \ket{a_U}$ where $\ket{a_U}$ is the output ancillary system that possibly depends on U.\\

To prove the theorem, it must be shown that no quantum gates $\mathcal{A}_n$ and $\mathcal{B}_n$ can be found such that for all unitary gates $\mathcal{U}$ 
\begin{align}
B_n(U\otimes\mathds{1})A_n\ \ket{\psi}\otimes \ket{0}=U^n\ket{\psi}\otimes\ket{a_U}. \label{eq:5}
\end{align}
This can be seen by noticing the linearity of the LHS of Eq. (\ref{eq:5}) with respect to $U$, while the RHS seems not to be so.
\begin{figure}
\centering 
\includegraphics[width=0.45\textwidth]{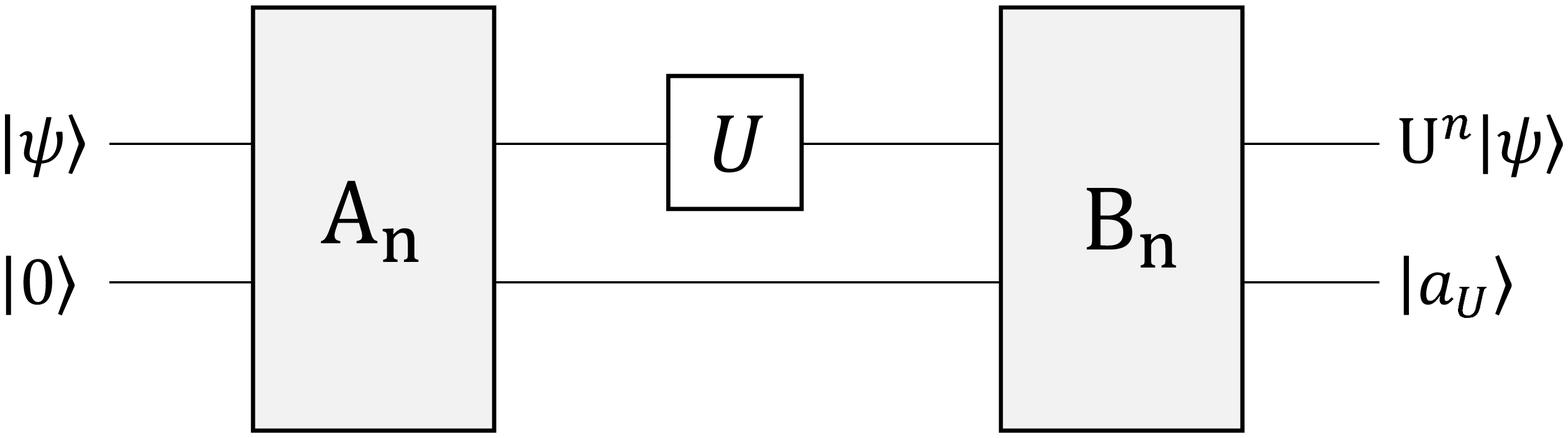}
\caption{The Stinespring realization of the quantum circuit of $\mathsf{Iter}_n$}
\label{fig:ckt} 
\end{figure}
In order to show that this is not possible, we use the linearity of quantum mechanics. To proceed, we need two unitary operators so that their linear combinations is also unitary. We take these operators to be $\mathds{1}$ and $\Omega=\sum_{k=0}^{d-1} \ket{k}\bra{d-k-1}= \left(\begin{smallmatrix}
• & • & 1 \\ 
• & \iddots & • \\ 
1 & • & •
\end{smallmatrix}\right) $.  Note that $\Omega$ is both unitary and Hermitian, so $\Omega^2=\mathds{1}$, which makes 
 $U:= \cos\theta \ \mathds{1}+i \sin\theta \ \Omega$, also unitary for every $\theta$.  Therefore, we should have
\begin{align}
B_n(\mathds{1}\otimes\mathds{1})A_n\ \ket{\psi}\otimes \ket{0} &=\mathds{1}\ket{\psi}\otimes \ket{a_{\mathds{1}}} \cr
B_n(\Omega\otimes\mathds{1})A_n\ \ket{\psi}\otimes\ket{0} &=\Omega^n \ket{\psi}\otimes \ket{a_{\Omega}} \\
B_n\left((\cos\theta \ \mathds{1} + i\sin\theta \ \Omega)\otimes\mathds{1}\right) A_n\ \ket{\psi}\otimes \ket{0} & =(\cos\theta \ \mathds{1} + i\sin\theta \ \Omega)^n\ket{\psi}\otimes \ket{a_{\theta}},\nonumber
\end{align}
where $\ket{a_{\theta}}$ stands for $\ket{a_{\cos\theta \mathds{1} + i\sin\theta \Omega}}$. Using the first two equations in the LHS of the third, we find 
\begin{align}
\cos \theta\ \mathds{1}\ket{\psi}\otimes \ket{a_{\mathds{1}}} + i \sin\theta\ \Omega^n \ket{\psi}\otimes \ket{a_{\Omega}} =\left(\cos n\theta\ \mathds{1} + i\sin n\theta\ \Omega\right) \ket{\psi}\otimes \ket{a_{\theta}}.
\end{align}

By looking at a specific entry of the first factor, we get
\begin{align}
\cos \theta\ \bra{0}\mathds{1} \ket{d-1}\ \ket{a_{\mathds{1}}}+ i \sin\theta\ \bra{0}\Omega^n
\ket{d-1}\ \ket{a_{\Omega}} =\bra{0}\left(\cos n\theta\ \mathds{1} + i\sin n  \theta\ \Omega\right) \ket{d-1}\  \ket{a_{\theta}},
\end{align}
but since $\Omega^n=\Omega$ for odd $n$, and $\Omega^n=\mathds{1}$ for even $n$, and $\bra{0}\mathds{1} \ket{d-1}=0$, $\bra{0} \Omega \ket{d-1}=1$, we find

\begin{align}
\sin n\theta\ =\left\{
\begin{array}{c l}      
    0& n~ even\\
   \pm \sin \theta\ & n~ odd
\end{array}\right.
\end{align}
that cannot be satisfied for arbitrary $\theta$.\\
\end{proof}

Although it is impossible to perfectly iterate an unknown gate in $\mathbb{U}(d)$, if a unitary is randomly picked from a set of \emph{jointly perfectly discriminable unitaries}, then clearly it is possible to iterate it. Whether or not the set of perfectly discriminable unitaries is the only set with this property, is an open question and remains for further investigation \cite{Referee_note}. \\

In the forthcoming sections, we explore different schemes that not perfectly but approximately bypass the introduced no-go result.
\section{The random guess strategy}
\label{sec:random-guess}
In this section, we investigate the random guess strategy in which the input gate is discarded and iterations of a randomly chosen unitary channel are applied to the input state. The random gate is chosen according to a probability distribution induced by normalized Haar measure on $\mathbb{U}(d)$. Therefore, it can be expressed with the following process
\begin{align}\label{eq:-16}
\mathcal{J}_n(\rho)&=\int dV\ V^n\rho V^{n \dagger}.
\end{align}
The motivation for studying this rather simple or blind strategy is that it plays an important role for understanding and comparing the performance of other strategies discussed in the following sections.\\

Let us begin by a theorem on the fidelity of this process:
 \begin{theorem}\label{thm:random-guess}
The fidelity of the random guess strategy is 
\begin{align}
F_{rand,n}=p_n^2+\frac{1-p_n^2}{d^2},
\end{align}
where
\begin{align}
p_n=\frac{min(n,d)-1}{d^2-1}.\label{eq:6}
\end{align}
\end{theorem}
Before we give the proof of this theorem, we need two lemmas.
\begin{lemma}\label{Lemma:1}
For all self-adjoint matrices $M\in L(\mathcal{H}_d)$ 
\begin{align}
\mathcal{J}_n(M):=\int dU\ U^n M U^{n \dagger}= p_n M +(1-p_n)\Tr(M)\frac{\mathds{1}}{d},
\end{align}
where the integration is with respect to the normalized Haar measure on $\mathbb{U}(d)$, and $p_n$ is the same as Eq. (\ref{eq:6}).
\end{lemma}

\begin{proof}[\textbf{Proof of Lemma}]
Let $\mathcal{E}$ be any quantum channel. The \emph{twirled transformation} associated with $\mathcal{E}$ is defined as
\begin{align}
\tilde{\mathcal E}(M):=\int dV\ V \mathcal{E} (V^{\dagger}M V) V^{\dagger}. \label{eq:7}
\end{align}
It is shown in in \cite{Emerson}, that  the twirled transformation acts like a depolarizing channel with parameter $p_{n,\mathcal{E}}$ that depends on the original channel $\mathcal{E}$:
 \begin{align}
\tilde{\mathcal{E}}(M)=p_{n,\mathcal{E}} M +(1-p_{n,\mathcal{E}})\Tr(M)\frac{\mathds{1}}{d}.
\end{align}

For the specific channel $\mathcal{J}_n$, the twirled channel $\tilde{\mathcal{J}_n}$ equals $\mathcal{J}_n$ itself. To see this, we note that 

\begin{align}
\tilde{\mathcal{J}}_n(\rho)=\int dV \int dU\ (VU^nV^{\dagger}) \rho (VU^{n \dagger}V^{\dagger})\label{eq:8}.
\end{align}
By defining $W:=VUV^{\dagger}$ and using right and left invariance of Haar measure, we find   
\begin{align}
\tilde{\mathcal{J}_n}(\rho)&=\int dV \int dW\ W^n \rho W^{n \dagger}=\int dV \mathcal{J}_n(\rho)=\mathcal{J}_n(\rho),
\end{align}
where we have used the normalization $\int dV=\mathds{1}$. It remains to determine the value of the parameter $p_n:=p_{n,\mathcal{J}_n}$. To do this, we enact the channel $\mathcal{J}_n$ on the matrix $|i\ra\la j|$ to  obtain
\be
\int dU\ U^n |i\ra \la j|  U^{n \dagger}= p_n |i\ra \la j| +(1-p_n)\delta_{ij}\frac{\mathds{1}}{d}.
\ee
Multiplying both sides by $\la i|$ and $|j\ra$ and summing over $i$ and $j$, we find
\be
\int dU |\Tr(U^n)|^2=p_n d^2+(1-p_n).
\ee

Using  theorem 2.1 of \cite{Diaconis}, according to which $\int dU |\Tr(U^n)|^2=min(n,d)$, we finally find the value of $p_n$:
\begin{align}
p_n=\frac{min(n,d)-1}{d^2-1}.
\end{align}
This completes the proof.
\end{proof}

\begin{corollary}\label{Lemma:2}
For all self-adjoint matrices $M \in L(\mathcal{H}^{(1)}_{d}\pr{\mathcal{H}^{(2)}_{d}})$
\begin{align}
\mathcal{K}_n(M):&=\int dU (U^n \pr{\mathds{1})}\ M\ (U^{n \dagger}\pr{\mathds{1}}) \label{eq:12} \nonumber \\
		  &=p_n M+(1-p_n)\frac{\mathds{1}\pr \Tr_1(M)}{d}.
\end{align}
\end{corollary}

\begin{corollary}\label{cor:1}
By substituting $M=\kett{\mathds{1}}\braa{\mathds{1}}$ in Eq. (\ref{eq:12}), we get
\begin{align}
\int dU\ \kett{U^n}\braa{U^n}=p_n \kett{\mathds{1}}\braa{\mathds{1}}+(1-p_n)\frac{\mathds{1}}{d^2}.\label{eq:-14}
\end{align}
\end{corollary}

\begin{proof}[\textbf{Proof of Theorem \ref{thm:random-guess}}]
According to Eq. (\ref{eq:-10}), approximating the iteration of a given gate $\mathcal{U}$ with the iteration of a Haar-distributed random gate $\mathcal{V}$ has the fidelity $\mathcal{F}(\mathcal{U}^n, \mathcal{V}^n)=|\brakett{U^n}{V^n}|^2$, so its expected value is
\begin{align}
\mathbb{E}\left[ \mathcal{F}(\mathcal{U}^n,\mathcal{V}^n) \right]&=\braa{U^n}\left(\int dU \kett{V^n}\braa{V^n}\right)\kett{U^{n}} \\
&=p_n\ |\brakett{U^n}{\mathds{1}}|^2+(1-p_n)\frac{1}{d^2},
\end{align}
where the second equality follows from Eq. (\ref{eq:-14}). The fidelity of the random guess strategy can be obtained by taking the average over all $\mathcal{U}$'s:
\begin{align}
F_{rand,n}&=\int dU\ \mathbb{E}\left[ \mathcal{F}(\mathcal{U}^n,\mathcal{V}^n) \right] \label{eq:-15}\\
								  &=p_n \int dU |\brakett{U^n}{\mathds{1}}|^2+(1-p_n)\frac{1}{d^2}.
\end{align}
Again, the integral is simplified using Eq. (\ref{eq:-14}):
\begin{align}
F_{rand,n} =p_n^2+\frac{1-p_n^2}{d^2}.
\end{align}
\end{proof}

As stated in Theorem \ref{thm:random-guess} and depicted in Fig. \ref{fig:plot1}, the fidelity of the random guess decreases quadratically with growth of the dimension $d$, but less intuitively is a kind of \emph{phase transition} occurs at $n=d$: the fidelity increases with growth of $n$ for $n<d$, and reaches a constant value for $n\geq d$. The random guess is not the only scheme with such phase transition, and as we shall see in next sections, this is a characteristic of all of our approximating strategies. The same phenomena is observed in a similar context when dealing with the joint distribution $f(\theta_1,\dots,\theta_d)$ of eigenvalues $\{e^{i \theta_j}\}_{j=1}^d$ of a Haar-distributed unitary matrix in $\mathbb{U}(d)$ \cite{Diaconis_2}, that is 
\begin{align}\label{eq:-9}
f(\theta_1,\dots,\theta_d)=\frac{1}{(2\pi)^d d!}\prod_{j<k} |e^{i\theta_j}-e^{i\theta_k}|^2,
\end{align}
so when $\theta_j\rightarrow \theta_k$, $f\rightarrow 0$, and eigenvalues somehow repel each other. To see this more intuitively, consider $d$ identically charged particles confined to move on the unit circle with Coulomb interaction between them. Their associated Gibbs distribution is
 \begin{align}
f(\theta_1,\dots,\theta_d)=\frac{1}{(2\pi)^d d!}e^{-\beta H(\theta_1,\dots,\theta_d)},
\end{align}
with the Hamiltonian $H=-\sum_{j<k} \log|e^{i\theta_j}-e^{i\theta_k}|$ and $\beta=2$. 
This is the same distribution as in Eq. (\ref{eq:-9}) and the repulsion of eigenvalues comes to have a clear physical meaning, and is similar to the repulsion between particles in the ordinary Coulomb gas. \\

When the joint distribution of eigenvalues of higher powers is considered, a phase transition occurs, and for $n\geq d$, the eigenvalues of $U^n$ are exactly distributed as $d$ points chosen independently and uniformly on the unit circle \cite{Rains}. Thus, the eigenvalues that seem to have an ordered structure and are very neatly spaced for $n=1$ have no structure for $n\geq d$.\\

To see the connection of this result to the fidelity of the random guess, notice that from the proof of Lemma \ref{Lemma:1}, the parameter $p_n$ in Eq. (\ref{eq:6}) is 
\begin{align}
p_n=\frac{\int dU\ |\Tr(U^n)|^2-1}{d^2-1},
\end{align}
but $\int dU\ |\Tr(U^n)|^2$ depends on the joint distribution of eigenvalues of $U^n$. For $n\geq d$, this distribution remains the same and $\int dU\ |\Tr(U^n)|^2=d$, so $p_n$ and fidelity also remain constant. \\

Finally, we prove that there exists a depolarizing channel whose fidelity equals the fidelity of the random guess and may be considered as an implementation of that.
\begin{theorem}\label{thm:ckt for random}
The fidelity of the random guess strategy for approximating the $n$'th iteration of a given unknown unitary $\mathcal{U}$ is equal to the fidelity of the following depolarizing channel
\begin{align}
\mathcal{J}_n(\rho)&=\int dV\ V^n\rho V^{n \dagger}\nonumber\\
&=p_n \rho +(1-p_n)\frac{\mathds{1}}{d},
\label{eq:20}
\end{align}
with $p_n=\frac{min(n,d)-1}{d^2-1}$, the same as in Eq. (\ref{eq:6}).
\end{theorem}
\begin{proof}[\textbf{Proof}]
Consider the quantum channel 
\begin{align}
\mathcal{J}_n(\rho)=\int dV\ V^n\rho V^{n \dagger},
\end{align}
with $V^n$ as its Kraus operators. The Choi operator of this channel is 
\begin{align}
R_{\mathcal{J}_n}=\int dV\ \kett{V^n}\braa{V^n},
\end{align}
so the fidelity of $\mathcal{J}_n$ is
\begin{align}
F(\mathcal{J}_n)&=\int dU  \braa{U^n}R_{\mathcal{J}_n}\kett{U^n}\nonumber\\
&=\int dU \int dV \ |\brakett{U^n}{V^n}|^2,
\end{align}
which is exactly the same as Eq. (\ref{eq:-15}), so
\begin{align}
F(\mathcal{J}_n)=F_{rand,n}.
\end{align}
It is also clear from Lemma \ref{Lemma:1} that $\mathcal{J}_n$ is a polarizing channel
\begin{align}
\mathcal{J}_n(\rho)=p_n\rho+(1-p_n)\frac{\mathds{1}}{d},
\end{align}
with $p_n$ as in Eq. (\ref{eq:6}).
\end{proof}
\section{The estimation strategy}\label{sec:quantum circuit of a repeater}
In the previous scheme, we blindly iterated a random gate and found its fidelity. We now discuss a more prepared and discriminating strategy in which we use more resources. Namely, we estimate the given unitary blackbox and based on the result of the estimation, we choose a gate and perform its iteration on the input state. \\

To make things clear, we can compare it with the random guess circuit that is described by the channel $\mathcal{J}_n(\rho)=\int dV\ V^n\rho V^{n \dagger}$, Eq. (\ref{eq:-16}). $\mathcal{J}_n(\rho)$ is the average of all states $V^n\rho V^{n \dagger}$, and each unitary gate $V$ has the same weight in it. We can use estimation results to give higher weights to more preferable gates. Let us denote the weight of unitary gate $V$ by $\omega_V$, so the action of our approximate gate iterator is to take $\rho$ and gives $\int dV \omega_V V^n\rho V^{n \dagger}$ as the output. As much as these weights are decided correctly, our circuit has higher fidelity than that of the random guess and approximates $\mathsf{Iter_n}$ better. \\

One idea for obtaining reasonable weights $\omega_V$ from a unitary channel with a single try is to encode the effect of the channel into a maximally entangled state and then perform a measurement on the state, the \emph{measure-and-prepare} method \cite{Estimation}. To see how this works, we first notice that according to Corollary \ref{cor:1}
\begin{align}
\int dV \kett{V}\braa{V}=\frac{\mathds{1}}{d}, \label{eq:-18}
\end{align} 
so the set of operators $d\ \kett{V}\braa{V}$ provides bases for a non-orthogonal measurement. On the other hand, the state $\kett{U}$ can be prepared by a single use of $\mathcal{U}$, Eq. (\ref{eq:-8}). Obviously, this measurement cannot perfectly discriminate $\kett{U}$ from other states, and after the measurement, a vector $\kett{V}$ is obtained with probability density $d^2|\brakett{V}{U}|^2$. Thus, the output state is a weighted mean of states $V^n\rho V^{n \dagger}$ with $\omega_V=d^2|\brakett{V}{U}|^2$. \\

As depicted in Fig. \ref{fig:ckt}, the circuit may include a preprocessing unit ($\mathcal{A}_n$) for preparing necessary states for estimation of $\mathcal{U}$, and a postprocessing unit ($\mathcal{B}_n$) for performing the measurement and preparing the output state based on the estimation. Let the input state on which the iteration is performed be $\rho$ and the ancillary system be a bipartite state $\ket{00}\bra{00}$. The preprocessing channel $\mathcal{A}_n$, prepares a maximally entangled state $\kett{\mathds{1}}\braa{\mathds{1}}$ from the ancillary system and swaps that with $\rho$, so that the input state remains unchanged until the estimation results are ready. In other words,
\begin{align}
\mathcal{A}_n(\rho\pr{\ket{00}\bra{00}})
												&=\kett{\mathds{1}}\braa{\mathds{1}}\pr{\rho}.
\end{align}
 Then, $\mathcal{U}\otimes \mathcal{I}$ acts on the entangled ancillary system and state $\kett{U}\braa{U}$ is prepared. In channel  $\mathcal{B}_n$, according to the results of the measurement of $\kett{U}\braa{U}$, unitary gates are performed on $\rho$, so the output is
\begin{align}
\mathcal{B}_n(\kett{U}\braa{U}\pr{\rho})&=\int dV (d^2|\brakett{V}{U}|^2)\  \kett{V}\braa{V}\pr{V^n \rho V^{n \dagger}}.
\end{align}
The action of the whole circuit on $\rho$ is obtained by tracing the output of $\mathcal{B}_n$ over the ancillary system:
\begin{align}\label{eq:19}
\mathcal{D}_{n,\mathcal{U}}(\rho)&:=\Tr_a\left(\mathcal{B}_n(\kett{U}\braa{U}\pr{\rho})\right)\nonumber\\
&=d^2 \int dV |\brakett{V}{U}|^2 \ V^n \rho V^{n \dagger}.
\end{align}
\begin{theorem}
The fidelity of this strategy is 
\begin{align}
F_{est,n}=d^2 \Tr \left(M_n^2\right), \label{eq:18}
\end{align}
with
\begin{align}
M_n:=\int dU\ \kett{U}\braa{U}\ \pr{\kett{U^n}\braa{U^n}}.
\end{align}
\end{theorem}

\begin{proof}[\textbf{Proof}]
The Choi operator associated with the map $\mathcal{D}_{n,\mathcal{U}}$ is 
\begin{align}
R_{\mathcal{D}_{n,\mathcal{U}}}=d^2\ \braa{U} \left(\int dV\ \kett{V}\braa{V}\ \pr{\kett{V^n}\braa{V^n}}\right) \kett{U}.
\end{align}
The fidelity of this strategy is $F_{est,n}=\int dU\ \Tr(R_{\mathcal{D}_{n,\mathcal{U}}}R_{\mathcal{U}^n})$. By replacing $R_{\mathcal{U}^n}=\kett{U^n}\braa{U^n}$, it is immediate to get Eq. (\ref{eq:18}).
\end{proof}

The matrix $M_n$ can be calculated numerically using Monte Carlo method. One approach is to generate Haar-distributed random unitary matrices in $\mathbb{U}(d)$. Then, the integral in Eq. (\ref{eq:18}) can be approximated by averaging the integrand over these random matrices. A simple algorithm exists for uniform generation of random matrices \cite{Random}. The idea is to generate a matrix $Z$ with \emph{QR-decomposition}
\begin{align}
Z=QR,
\end{align}
where $Q$ is unitary and $R$ is upper-triangular and invertible. Let  $D$ be the diagonalization of $R$ whose entries are divided by their absolute value, then it turns out that if entries of $Z$ are \emph{i.i.d} standard complex normal random variables, the matrix $U=QD$ is distributed according to Haar measure \cite{Random}. \\

The fidelity of the estimation strategy for different iteration orders is depicted in Fig. \ref{fig:plot1}. It can be seen that by increasing the dimension of the input system, fidelity of the proposed circuit decreases and performance of this circuit tends to that of the random guess method. Similar to the case of the random guess, the fidelity of this circuit reaches a constant value and remains the same for higher order iterations.\\
\begin{figure}
 
\includegraphics[width=8.6 cm]{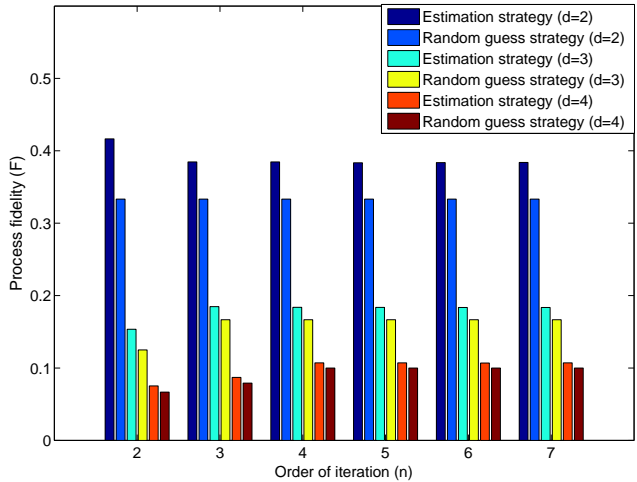}
\caption{
Fidelity of the random guess and the estimation strategies for various orders of iteration in dimensions $d=2, 3$ and $4$. The advantage of estimation strategy over random guess is clear. It is also seen that this advantage tends to decrease with increasing the dimension $d$.
}
\label{fig:plot1} 
\end{figure}

Note that in the estimation strategy we could have replaced the measurement in the over-complete basis in Eq. (\ref{eq:-18}), by a measurement over an orthonormal basis
\begin{align}\label{eq:-19}
&\brakett{U_j}{U_k}=\Tr (U^{\dagger}_j U_k)=0 \quad \forall j, k\in\{1,\dots,d^2\},\ j\neq k, \nonumber \\
&\sum_{j=1}^{d^2} \kett{U_j}\braa{U_j}=\mathds{1}.
\end{align} 
which is a basis of jointly perfectly discriminable gates. In this case, the quantum channel (\ref{eq:19}) would have been replaced by 
\begin{align}
\tilde{\mathcal{D}}_{n,\mathcal{U}}(\rho)=\sum_{j=1}^{d^2} |\brakett{U_j}{U}|^2 \ U_j^n \rho U_j^{n \dagger}.
\end{align}
However, as we will show in the next section, none of these two kinds of estimation are optimal. The same is true in the case of cloning of unitary gates \cite{unitary-cloning}, where $F_{est}$ is even worse than $F_{rand}$ for $d>2$. However, for $n=-1$, i.e, when the unknown gate is to be inverted, the estimation strategy is the optimal scheme \cite{Inverse}. \\
\section{The identity and direct channels strategies}\label{sec:idn strategies}
An apparently trivial strategy, called the \emph{identity channel strategy}, is to take the identity channel as an approximation of $\mathsf{Iter}_n$, i.e., to neglect the given gate $\mathcal{U}$, and to put the input state $\rho$ directly in the output. To see why this approximation is \emph{reasonable}, we note that by using Eq. (\ref{eq:-10}) and Corollary \ref{cor:1}, we get
\begin{align}\label{eq:-11}
\int dU \mathcal{F}(\mathcal{V},\mathcal{U}^n)=p_n|\brakett{V}{\mathds{1}}|^2+\frac{1-p_n}{d^2},
\end{align}
which immediately gives
 \begin{align}
\max_{V\in \mathbb{U}(d)}  \int dU \mathcal{F}(\mathcal{V},\mathcal{U}^n)=\mathcal{F}(\mathds{1},\mathcal{U}^n)
\end{align}
so on average, the identity channel has the maximum similarity to all $\mathcal{U}^n$'s, and in this sense, it is a reasonable approximation of $\mathcal{U}^n$.\\

The fidelity of this channel is obtained by replacing $\mathcal{V}$ with $\mathds{1}$ in Eq. (\ref{eq:-11}), that gives:
\begin{theorem}
The fidelity of the identity channel $F_{iden,n}$ is 
\begin{align}
F_{iden,n}&=\int dU \mathcal{F}(\mathcal{I},\mathcal{U}^n)\nonumber\\
&=p_n+\frac{1-p_n}{d^2}
.
\end{align}
where $p_n$ is given by Eq. (\ref{eq:6}).
\end{theorem}
As it can be seen in Fig. \ref{fig:plot_data_2}, the performance of this process is better than the estimation method, and for $n\geq d$, $F_{idn,n}=\frac{1}{d}$. In fact, we will see in the next section that this channel achieves the optimum fidelity in certain cases. \\

The \emph{direct channel strategy} is another trivial method with the similar performance for high enough orders n. In this case, $\mathcal{U}^n$ is approximated by $\mathcal{U}$ and the given gate is performed \emph{directly} on the input state by replacing $\mathcal{A}_n$ and $\mathcal{B}_n$ with identity channels. Numerical results for this scheme is depicted in Fig. \ref{fig:plot_data_2}. Note that the same phase transition as in case of the estimation and random guess strategies occurs. \\

\begin{figure}
\centering 
\includegraphics[width=8.6 cm]{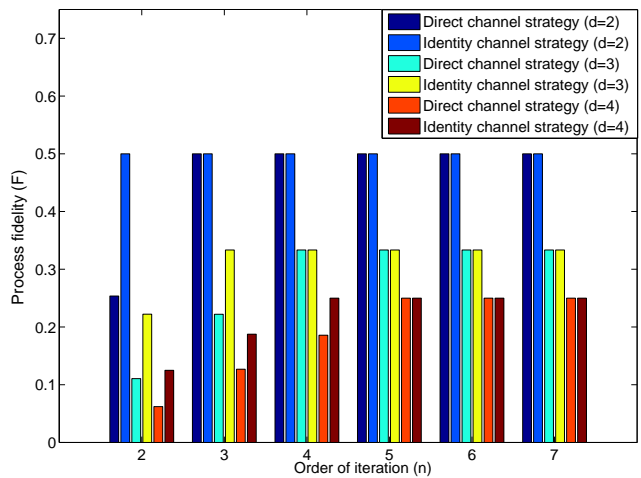}
\caption{Fidelity of the identity and direct channel strategies. Both schemes have equal performances for $n>d$, and, in fact, they achieve the optimum fidelity in this case.}
\label{fig:plot_data_2} 
\end{figure}
\section{Optimum fidelity}\label{sec:Optimal strategies}
The most general form of the quantum circuit of $\mathsf{Iter}_n$ may be described by concatenation of different unitary channels and some ancillary systems, namely, the Stinespring realization shown in Fig. \ref{fig:ckt}. Therefore, one way to find the optimal process that faithfully realizes $\mathsf{Iter}_n$ is to maximize the fidelity over all quantum channels $\mathcal{A}_n$ and $\mathcal{B}_n$. This is not the only way, and in fact, a more suitable way to describe $\mathsf{Iter}_n$ exists: the quantum comb notion \cite{Arch}.\\

In this method, instead of considering separate channels $\mathcal{A}_n$ and $\mathcal{B}_n$, they are merged and replaced with a channel $\mathcal{C}_n$ from $D(\mathcal{H}^{(0)})\pr{D(\mathcal{H}^{(2)})}$ to $D(\mathcal{H}^{(1)})\pr{D(\mathcal{H}^{(3)})}$, see Fig. \ref{fig:Picture4}. This channel has an open slot in which the given unitary gate $\mathcal{U}$ is inserted and the $n$'th iteration of $\mathcal{U}$ is realized, Fig \ref{fig:Picture5}. \\

\begin{figure}
\hspace*{.1 cm}
\subfloat[]{\label{fig:Picture4}}%
  \includegraphics[width=0.3\textwidth]{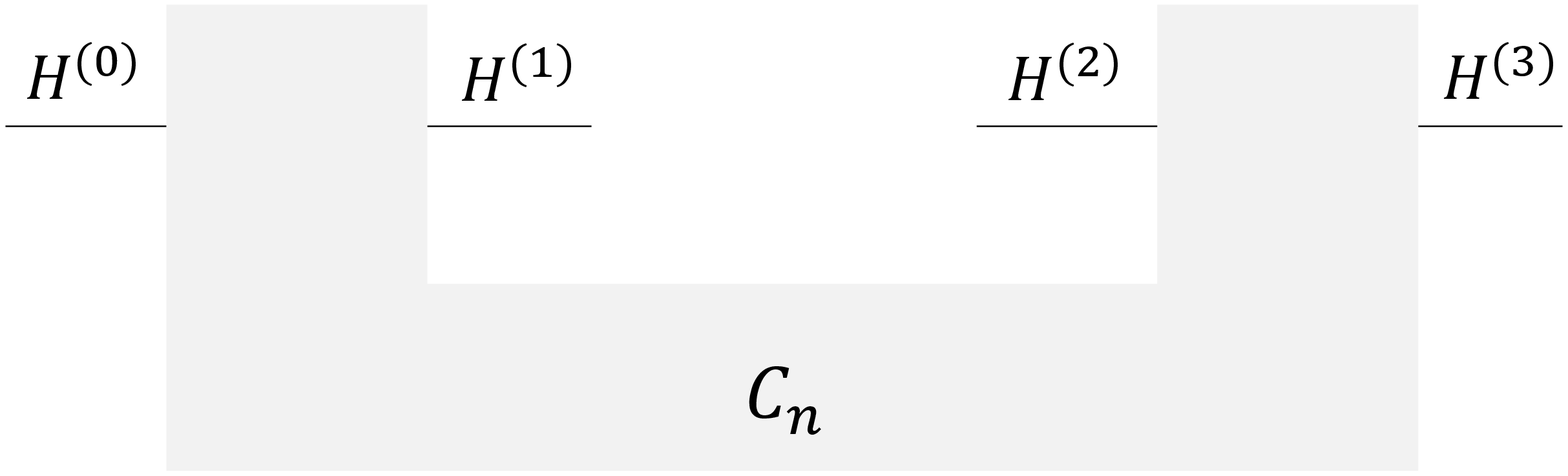}%
\hspace{.5 cm}
\subfloat[]{ \label{fig:Picture5}}%
  \includegraphics[width=0.35 \textwidth]{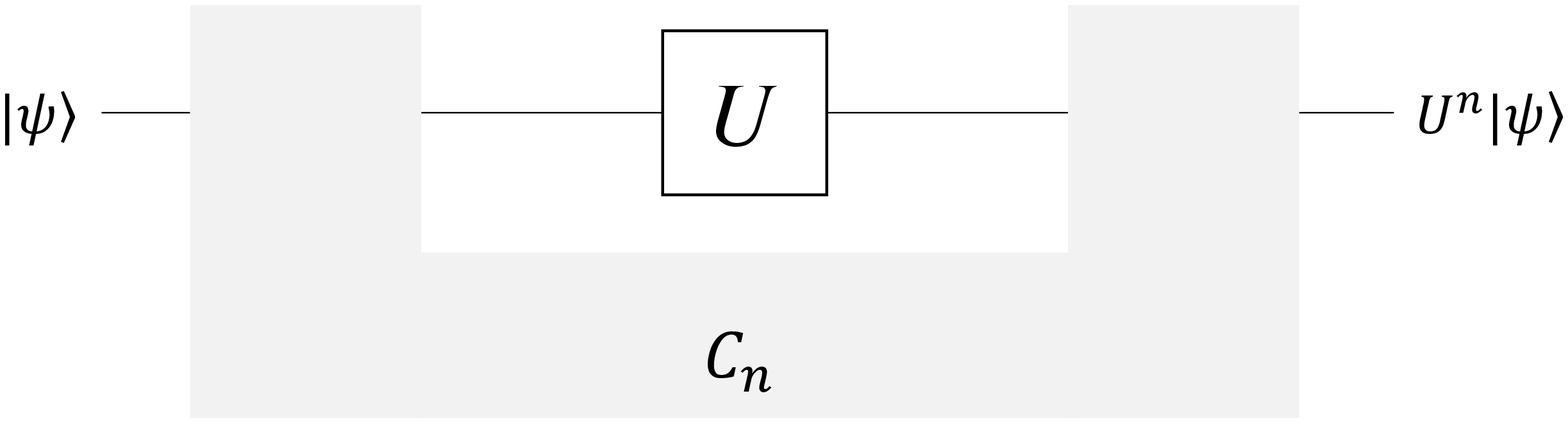}%

  \caption{(a) The channel $\mathcal{C}_n$ with one open slot is a replacement for the Stinespring realization shown in Fig. \ref{fig:ckt}, (b) the unitary gate is inserted into the slot to realize $\mathsf{Iter}_n$.}
\end{figure}

In the circuit shown in Fig. \ref{fig:Picture5}, the processing of information is from left to right as time passes, and the outputs may depend on the previous but not the later times inputs. Thus, not every quantum channel from $D(\mathcal{H}^{(0)})\pr{D(\mathcal{H}^{(2)})}$ to $D(\mathcal{H}^{(1)})\pr{D(\mathcal{H}^{(3)})}$ can be realized with such a circuit, and they need to meet additional \emph{causality constraints}. \\

The \textit{quantum comb} $R_{\mathcal{C}_n}$ is defined as the Choi operator associated with $\mathcal{C}_n$ acting on $D(H^{(1)})\pr{D(H^{(3)})}\pr{D(H^{(0)})}\pr{D(H^{(2)})}$. As it is proven in \cite{Arch}, the causality constraint is equivalent to the following set of linear constraints on the quantum comb $R_{\mathcal{C}_n}$
\begin{align}
\Tr_3(R_{\mathcal{C}_n})&=\frac{\mathds{1}_2}{d}\pr{R_{\mathcal{C}_n}^{(1)}} \nonumber\\
\Tr_1(R_{\mathcal{C}_n}^{(1)})&=\frac{\mathds{1}_0}{d}, \label{eq:-12}
\end{align}
where $\Tr_i$ means the partial trace over $H^{(i)}$ and $R_{\mathcal{C}}^{(1)}$ is a Choi operator on $\mathcal{H}^{(1)}\pr{\mathcal{H}^{(0)}}$, and subscripts of operators represent the related Hilbert space of them.\\

The benefit of using the quantum comb notion is clearly seen when the composition of $\mathcal{C}_n$ with the unitary gate $\mathcal{U}$, denoted by $\mathcal{C}_n \star \mathcal{U}$, is to be described, Fig. \ref{fig:Picture5}. It can be proven that (see Ref. \cite{Arch}) the Choi operator of the channel $\mathcal{C}_n \star \mathcal{U}$ is :
\begin{align}
R_{\mathcal{C}_n \star \mathcal{U}}=d^2\braa{U^*_{21}}R_{\mathcal{C}_n}\kett{U^*_{21}},
\end{align}
where $R_{\mathcal{C}_n \star \mathcal{U}}\in L(\mathcal{H}^{(3)}\otimes \mathcal{H}^{(0)})$ and $U^*$ is the conjugate complex of $U$. The subscripts $21$ in $\kett{U^*_{21}}$ denotes the domain and image Hilbert spaces of the operator $U$. Thus, the fidelity (Eq. (\ref{eq:-17})) is 
\begin{align}
F(\mathcal{C}_n) &=\int dU\ \Tr( d^2 \braa{U^*_{21}} R_{\mathcal{C}_n} \kett{U^*_{21}}\ \kett{U^n_{30}}\braa{U^n_{30}}\ )\nonumber\\
&=\Tr(d^2 R_{\mathcal{C}_n} \int dU\ \kett{U^n_{30}}\braa{U^n_{30}}\ \pr{\kett{U^*_{21}}\braa{U^*_{21}}}\ ).
\end{align}
Let $\tilde{M}_n:=d^2 \int dU\ \kett{U^n_{30}}\braa{U^n_{30}}\ \pr{\kett{U^*_{21}}\braa{U^*_{21}}}$, then 
\begin{align}
F(\mathcal{C}_n)=\Tr(R_{\mathcal{C}_n}\tilde{M}_n). \label{eq:-13}
\end{align}
Therefore, to find optimal strategies for realizing $\mathsf{Iter}_n$, the following optimization problem should be solved:
\begin{align}
\max_{R_{\mathcal{C}_n}}\ & \Tr(R_{\mathcal{C}_n}\tilde{M}_n)\nonumber \\
\sth \ &\Tr_3(R_{\mathcal{C}_n})=\frac{\mathds{1}_2}{d}\pr{R_{\mathcal{C}_n}^{(1)}}, \\
&\Tr_1(R_{{\mathcal{C}_n}}^{(1)})=\frac{\mathds{1}_0}{d},\nonumber \\
&R_{\mathcal{C}_n}\geq0,\ R_{\mathcal{C}_n}^{(1)}\geq 0\nonumber.
\end{align}
This is an example of Semidefinite Programming (SDP) \cite{Boyd}, which is numerically solvable using packages like CVX \cite{CVX}. The optimum fidelity obtained by this method and fidelity of the identity channel are shown for different cases in Fig. \ref{fig:data_plot_3}. For $d=2$, $n>2$, the identity and direct channels discussed in Sec. \ref{sec:idn strategies} achieve the optimum fidelity, and this is quite unanticipated, since both are trivial methods where resources like entanglement or general preprocessing or postprocessing units are not used. \\

As in the case of other approximating processes investigated earlier, the optimum fidelity reaches a constant value, and the optimal iterator has the same performance for high enough orders $n$. This phenomena is not observed in the case of $1$-to-$n$ cloning of unitary gates where the fidelity seems to decreases monotonically with growth of $n$ \cite{unitary-cloning}. In addition, in that problem, the performance of the optimal cloner depends crucially on the entanglement of input states with the ancillary system, and the identity channel has a by-far-worse fidelity than the optimal cloner.\\  

\begin{figure}
\centering 
\includegraphics[width=8.6 cm]{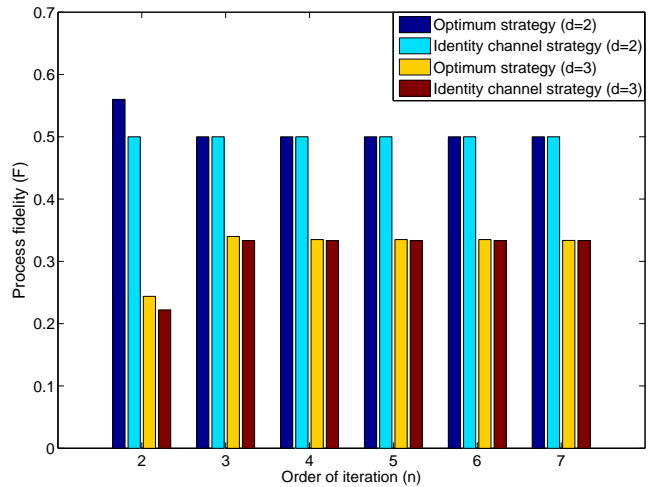}
\caption{The optimum fidelity in approximating $\mathsf{Iter}_n$.}
\label{fig:data_plot_3} 
\end{figure}
\section{Conclusion}\label{sec:Conclusion}

We have shown that it is impossible to iterate an unknown quantum gate by using it once, what we called it the no-iteration theorem. We have also investigated different schemes to approximately bypass this no-go result: (1) The random guess strategy in which iterations of a randomly chosen gate is performed. (2) The measure-and-prepare method where the given gate $\mathcal{U}$ is first estimated using the state $\kett{U}$, and then unitary processes are performed on the input state accordingly. (3) Approximating with the identity channel or by performing the given unitary process directly on the input system. In addition, by using the notion of quantum comb, we have been able to state the problem of finding the optimal iterator as a semidefinite programming, which we have solved numerically for $d=2, 3$.\\

The iteration problem has some unique features that make it different from similar problems like cloning of unitary channels. One is that the performance of all discussed methods including the optimal one, remains the same for highly enough orders $n$. In the case of random guess, we saw the connection of this phase transition to the joint distribution of eigenvalues of a random unitary matrix, which changes from being highly ordered to having no structure for $n\geq d$. The other feature is that the performance of trivial processes like identity or direct channels is comparable to the optimal strategies, and at least for $d=2,3$, numerical solutions show they achieve the optimum performance for $n >d$.\\

This no-go theorem is another example of transformations of quantum channels that cannot be realized perfectly. Providing these examples helps us to understand the characteristics of quantum operations as \emph{carriers of information}, and shows us how laws of quantum mechanics act when evolution of operations is considered instead of states.\\

Interesting behaviors of gate iterators discussed in this paper, motivates a more general study of powers of unitary operators in $n \gg 1$ limit. The iteration problem and the performance of the optimal iterator might also be explored when multiple copies of the oracle is provided.
\appendix	
\section{Alternate Proof of Theorem \ref{thm: no-rep}}\label{app:-1}

\begin{proof}[\textbf{Alternate proof}]
According to the following lemma, proved as a theorem in \cite{bounds_search}, there exists a lower bound on the performance of quantum search algorithms. This lower bound is only a few percent smaller than the number of iterations required by Grover’s algorithm \cite{Grover}.  
\begin{lemma}
Let $T$ be any set of $N$ strings, and $M$ be any oracle quantum machine with bounded error probability. Let $y \in_R T$ be a randomly and uniformly chosen element from $T$. Put $\mathcal{O}$ to be the oracle where $\mathcal{O}(x) = 1$ if and only if $x = y$. Then the expected number of times $M$ must query $\mathcal{O}$ in order to determine $y$ with probability at least $\frac{1}{2}$ is at least $\lfloor\sin(\frac{\pi}{8})\sqrt{N}\rfloor$.
\end{lemma}
Now imagine that $\mathsf{Iter_n}$ can be constructed perfectly, then for an appropriate number of strings $N$, the required number of queries can be reduced. This can be done easily by replacing each $n$ successive queries in Grover search algorithm with a single use of $\mathsf{Iter_n}$. Thus, the lower bound of the last Lemma is violated and this is a contradiction.
\end{proof}  



\end{document}